\documentclass[manuscript,screen,nonacm]{acmart}
\usepackage{amsmath,amsfonts,braket}
\usepackage{algorithmic}
\usepackage{algorithm}
\usepackage{threeparttable}
\usepackage{array}
\usepackage[caption=false,font=normalsize,labelfont=sf,textfont=sf]{subfig}
\usepackage{textcomp}
\usepackage{stfloats}
\usepackage{listings}
\usepackage{multirow}
\usepackage{pgfplots}
\pgfplotsset{compat=1.17}
\usepackage{url}
\usepackage{verbatim}
\usepackage{graphicx}
\usepackage[english]{babel}
\usepackage{amsthm}

\AtBeginDocument{%
  \providecommand\BibTeX{{%
    \normalfont B\kern-0.5em{\scshape i\kern-0.25em b}\kern-0.8em\TeX}}}

\begin{document}

\title[qIoV: A Quantum-Driven Internet-of-Vehicles-Based Approach]{qIoV: A Quantum-Driven Internet-of-Vehicles-Based Approach for Environmental Monitoring and Rapid Response Systems}

\author{Ankur Nahar}
\email{nahar.1@iitj.ac.in}
\orcid{0000-0001-9996-1167}
\affiliation{%
  \institution{Indian Institute of Technology Jodhpur}
  \streetaddress{Department of Computer Science and Engineering}
  \city{Jodhpur}
  \state{Rajasthan}
  \country{India}
  \postcode{342030}
}

\author{Koustav Kumar Mondal}
\email{mondal.4@iitj.ac.in}
\orcid{0000-0003-1679-1030}
\affiliation{%
  \institution{Indian Institute of Technology Jodhpur}
  \streetaddress{Inter-disciplinary Research Division - IoT and Application}
  \city{Jodhpur}
  \state{Rajasthan}
  \country{India}
  \postcode{342030}
}

\author{Debasis Das}
\email{debasis@iitj.ac.in}
\orcid{0000-0001-6205-4096}
\affiliation{%
  \institution{Indian Institute of Technology Jodhpur}
  \streetaddress{Department of Computer Science and Engineering}
  \city{Jodhpur}
  \state{Rajasthan}
  \country{India}
  \postcode{342030}
}

\author{Rajkumar Buyya}
\email{rbuyya@unimelb.edu.au}
\orcid{0000-0001-9754-6496}
\affiliation{%
  \institution{University of Melbourne}
  \streetaddress{Computing and Information Systems}
  \city{Melbourn}
  \state{VIC}
  \country{Australia}
  \postcode{3052}
}

\renewcommand{\shortauthors}{Nahar, et al.}

\begin{abstract}
This research addresses the critical necessity for advanced rapid response operations in managing a spectrum of environmental hazards. These include leaks in urban pipelines, industrial release of hazardous gases, methane emissions from landfill sites, chlorine gas leaks from water treatment facilities, and carbon monoxide emissions from residential areas, all of which are critical to ensuring public safety and environmental protection. We propose a novel framework, \textit{qIoV} that integrates quantum computing with the Internet-of-Vehicles (IoV) to leverage the computational efficiency, parallelism, and entanglement properties of quantum mechanics. Our approach involves the use of environmental sensors mounted on vehicles for precise air quality assessment. These sensors are designed to be highly sensitive and accurate, leveraging the principles of quantum mechanics to detect and measure environmental parameters. A salient feature of our proposal is the Quantum Mesh Network Fabric (QMF), a system designed to dynamically adjust the quantum network topology in accordance with vehicular movements. This capability is critical to maintaining the integrity of quantum states against environmental and vehicular disturbances, thereby ensuring reliable data transmission and processing. Moreover, our methodology is further augmented by the incorporation of a variational quantum classifier (VQC) with advanced quantum entanglement techniques. This integration offers a significant reduction in latency for hazard alert transmission, thus enabling expedited communication of crucial data to emergency response teams and the public. Our study on the IBM OpenQSAM 3 platform, utilizing a 127 Qubit system, revealed significant advancements in pair plot analysis, achieving over 90\% in precision, recall, and F1-Score metrics and an 83\% increase in the speed of toxic gas detection compared to conventional methods.Additionally, theoretical analyses validate the efficiency of quantum rotation, teleportation protocols, and the fidelity of quantum entanglement, further underscoring the potential of quantum computing in enhancing analytical performance.
\end{abstract}


\maketitle

\section{Introduction}
Gas leaks in urban areas often lead to highly dangerous situations and can potentially cause severe harm. The absence of a rapid response system has been a contributing factor in numerous incidents, such as the over 2,600 gas pipeline leaks in the United States from 2010 to late 2021 \cite{phmsa2023pipeline}, and 3,695 fatalities in China's coal mining industry from 2001 to 2018 \cite{zhou2022variance}. Additionally, from 2010 to 2022, there were 4,901 oil and gas spills in the U.S., with Texas alone accounting for nearly 40\% of these incidents \cite{visualcapitalist2023spills}. The Bhopal gas leak tragedy in India \cite{mishra2009bhopal}, resulting in thousands of deaths, underscores the need for efficient, real-time alert systems. Moreover, the U.S. pipeline regulator's recent measures to curb methane leaks from the country's vast gas infrastructure highlight the ongoing challenges and urgency in addressing such hazardous leaks \cite{reuters2023methaneleaks}.

In such scenarios, early warning systems are essential in reducing environmental disaster impacts \cite{balis2013development}. They enable fast evacuations, timely medical aid, and quick emergency service mobilization, which is particularly critical in gas leak scenarios. Classical computing methods in Hazardous Incident Tracking, such as wireless sensor networks \cite{dong2019gas} and machine learning algorithms \cite{liu2021predictive}, have made significant contributions in detecting and localizing gas leaks. However, they encounter limitations in real-time data processing, scalability, and sensitivity to subtle environmental changes. These limitations underscore the need for a more accurate and timely detection of hazardous incidents. Here, the quantum computing model can provide real-time, high-fidelity analysis of environmental data, enabling more accurate and timely detection of hazardous incidents compared to existing classical methods. The transition to quantum computing not only promises a significant speedup in processing times but also introduces a new paradigm in environmental monitoring, where the ability to detect minute changes in sensor data can lead to earlier warnings and more effective mitigation of hazards. Quantum computing application within vehicular networks presents a novel front with potential for enhancing safety and operational efficiency. This research points at the intersection of quantum computing and vehicular technology, aiming to leverage the computational capabilities of quantum systems to address critical challenges in vehicular sensor data processing and communication. Central to our investigation are three interconnected challenges: the translation of classical sensor data into quantum formats, the utilization of quantum classifiers for the detection of environmental hazards, and the deployment of quantum entanglement mechanisms for the dissemination of alerts within a vehicular network. By addressing these challenges, this study aims to contribute to the development of quantum computing applications that can surpass the limitations of classical computing paradigms. In this context, vehicles, with their omnipresence in urban areas, emerge as innovative carriers for environmental sensors (e.g., MQ2, MQ3, MQ5, MQ6, MQ7, MQ8, and MQ135), transforming them into mobile monitoring units. These sensors gather real-time data, which is crucial for hazard detection. Nevertheless, challenges arise due to this data's sheer volume and complexity, making traditional processing methods insufficient \cite{pramanik2022optimization}. Our framework in Fig. \ref{intr} presents a sustainable approach to environmental safety, leveraging a quantum classifier, quantum mesh fabric, and quantum server for enhanced detection and communication of hazardous gas leaks. The process begins with Step 1, where vehicles equipped with advanced sensors detect toxic gases emanating from sources such as industrial sites or urban pipelines. These sensors gather high-resolution environmental data crucial for the early detection of hazardous gas leaks. In Step 2, this sensor data is converted into a quantum-compatible format. This quantum data conversion is a critical process that translates classical sensor readings into quantum states that can be processed by quantum algorithms. Step 3 involves the core quantum computing components of the system—Quantum Mesh Fabric and Quantum Server. The Quantum Mesh Fabric is an adaptive network that manages the entanglement of quantum states across the vehicular network, ensuring robust and secure communication channels. The Quantum Server, on the other hand, is tasked with processing the quantum data using quantum algorithms, notably the variational quantum classifier, which identifies and classifies the detected gases based on their quantum state representations. Finally, step 4 is where the processed data culminates into actionable intelligence. The Quantum Server, utilizing the quantum entanglement established by the Quantum Mesh Fabric, disseminates hazard alerts across the network. These alerts are transmitted with minimal latency using the principles of quantum entanglement, which allow for the instantaneous relay of information. 

\begin{figure}[ht]
    \centering
    \includegraphics[width=0.9\columnwidth]{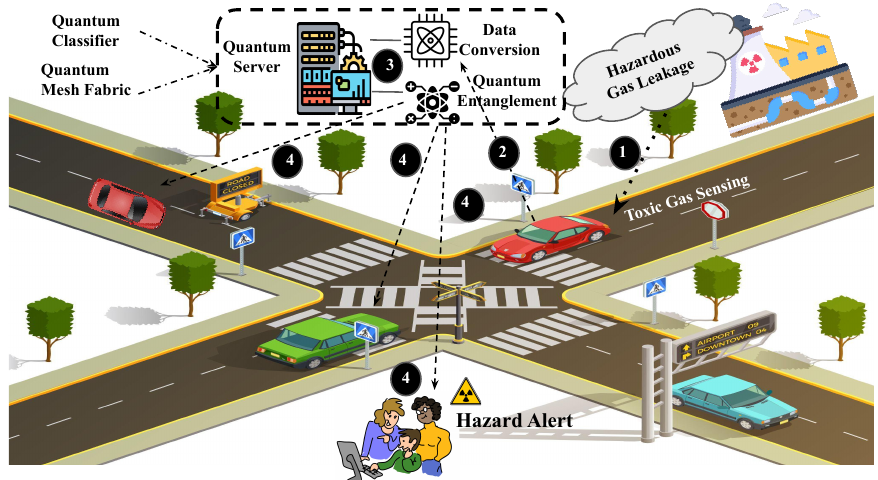}
    \caption{qIoV: A Quantum Framework for Toxic Gas Detection}
    \label{intr}
\end{figure}

\subsection{Motivation}
Serial incidents of disastrous environmental emergencies and recurring toxic gas breakouts in the world motivate our research. A report in \cite{india2006primer} claims that most tragedies occurring at nocturnal timing laid bare a critical vulnerability. The existing systems fail to effectively detect, predict, and warn against nocturnal toxic hazards due to limited sensor sensitivity and reduced human monitoring during night hours. For example, the Bhopal gas tragedy in India, a night-time industrial failure resulting in immediate and widespread loss of life \cite{palazzi2015critical}. Similarly, the Philadelphia refinery explosion in 2019 \cite{auchincloss2019re}, and China’s toxic gas breakout incidents echoed this necessity, emphasizing the elevated risks at a time of natural human vulnerability. It underscored the need for a vigilant, around-the-clock monitoring system. One promising solution could be public and private vehicles as emergency monitoring units. 

\textbf{\textit{However,  why vehicles, one might ask?}} The selection of vehicles as the primary medium for this research is underpinned by their ubiquity and inherent mobility. Given their extensive coverage and constant presence in urban areas, vehicles offer a unique platform for deploying environmental sensing capabilities. This ubiquity positions them as ideal candidates for monitoring a range of environmental parameters, including toxic gas concentrations, on a scale and with a granularity that stationary sensors cannot match \cite{cooper2021exploring}. However, the effective realization of this solution is contingent upon overcoming significant computational challenges. Traditional computing infrastructures, while capable in many respects, fall short in terms of processing speed and sensitivity, particularly when rapid data analysis is crucial—for instance, during nocturnal hours when the absence of human observers necessitates improved sensor vigilance. 

\textbf{\textit{How technology will help?}} Quantum computing introduces data processing capabilities, primarily through its exploitation of quantum bits (qubits). These qubits can exist in multiple states simultaneously (superposition) and be entangled with other qubits, creating a complex web of probabilistic relationships that can be manipulated collectively. This quantum superposition and entanglement enable the parallel processing of information, allowing quantum systems to analyze vast datasets and execute algorithms at speeds unattainable by traditional computing systems. In the context of vehicular networks, this quantum advantage translates into the ability to rapidly encode, process, and interpret the environmental data collected by vehicular sensors. Through quantum state encoding, sensor data can be transformed into quantum states, which are then manipulated using quantum circuits that execute sophisticated algorithms for data analysis \cite{danba2022toward}. 

\subsection{Contributions}

This research proposes a framework called \textit{qIoV} for environmental monitoring (Ref. Fig. \ref{intr}), particularly in the detection and analysis of toxic gasses. Our framework utilizes multipartite entanglement that refers to a quantum state involving three or more qubits where the state of each qubit cannot be described independently of the state of the others, even when the qubits are separated by large distances. Using these entangled states, qIoV framework encode environmental data, including hazard alerts, into a quantum state shared across multiple vehicles within the network. The transmission of alerts is further refined through the application of quantum teleportation technique. For vehicular networks, this implies that once an alert is encoded into a quantum state, it can be instantaneously transmitted across the network, from the detecting vehicle to others, utilizing entanglement patterns designed for vehicular data. To achieve this, we design quantum circuits that are capable of generating and manipulating multipartite entangled states, alongside protocols for entanglement distribution and quantum teleportation. The circuits include operations i.e., Bell state measurements and controlled quantum gates, which facilitate the reliable encoding and teleportation of quantum states corresponding to environmental alerts. The adaptability of these entangled states to the topology of vehicular networks ensures that alerts are propagated efficiently, reaching all intended recipients with minimal delay and maximal fidelity. To summarize, followings are the key contributions of our research:
\begin{enumerate}
    \item We propose an encoding scheme that aligns the multi-dimensional sensor data from vehicular sensors with quantum states. Our framework, specifically calibrated for detecting various toxic gases, relies on optimizing the rotation angles $\mathbf{R_\theta}$ based on the unique patterns of gas sensor readings. 
    \item  We propose multipartite entanglement in vehicular networks based on Quantum Entangled Feature Maps. The quantum operations systematically entangle qubits across different vehicles, utilizing controlled-Z (CZ) and controlled-X (CX) gates to induce correlations that mirror the interconnected data patterns of vehicular sensors.
    \item To accommodate the dynamic and potentially large-scale nature of vehicular networks, we structure the entanglement in layers. Each layer represents a level of interaction among the vehicles. This allows for scalable and flexible entanglement patterns that can be adapted based on the network's size and complexity.
    \item We design Adaptive Quantum Circuits to integrate entanglement patterns to vehicular data. This methodology involves the dynamic adjustment of quantum circuits to optimize how environmental alerts are encoded into the quantum states, ensuring efficient and accurate data transmission.
    \item Our framework applies real-time adaptive quantum circuit optimization, guided by error syndromes and decoherence patterns identified through stabilizer codes, to precisely adjust quantum gate parameters, ensuring robust error correction and maintaining high fidelity in quantum state transmissions across vehicular networks.
\end{enumerate}

This article is organized as follows: Section 2 discusses the literature on quantum computing in computation, Section 3 presents the problem definition, and Section 4 defines the system overview and architectural design. Section 5 proposes the \textit{qIoV} framework. Finally, Section 6 presents the simulation findings and testing, and we conclude this paper in Section 7.

\section{Related Work}

This section discusses the evolving landscape of quantum computing, highlighting related works that contribute to our understanding of quantum machine learning, its applications in transportation, and emerging quantum frameworks and protocols. It also presents insights into quantum technology's current state, challenges, and innovations to set a foundational context for our research.

\subsection{Quantum Computing Foundations and Data Networks}
The convergence of machine learning (ML) algorithms and quantum computing has been identified as a transformative approach for processing and analyzing extensive datasets with high-dimensional complexity. Research in \cite{ramezani2020machine} examined this integration, underscoring the paradigm-shifting potential of qubits to enhance computational efficiency and data analysis capabilities. Furthermore, the authored discussed the development of hybrid quantum-classical systems, designed to exploit quantum computational advantages while addressing prevalent challenges such as the scarcity of qubits and the omnipresent issues of noise within quantum systems. Notably, the authors acknowledged the scalability constraints and noise inherent to quantum computing frameworks. Subsequently, in \cite{yang2023survey}, authors illuminated significant progress in the domain of quantum hardware, which now supports configurations scalable to thousands of qubits, thereby facilitating an era of unprecedented computational speed and efficiency. Despite these technological advancements, the authors identified enduring challenges, including the instability of quantum systems, resource limitations, communication barriers, and vulnerabilities in security protocols. Quantum machine learning and its intricacies are explored in \cite{houssein2022machine,danba2022toward}. The authors discussed the state of the field amid the limitations of quantum hardware and algorithm complexity. 

The authors in \cite{fan2022hybrid} introduced a framework that integrated quantum and classical computing methodologies to surmount intricate network resource optimization challenges. The authors applied quantum annealing to solve pure integer programming problems, while ancillary challenges are addressed through classical computing techniques. However, the approach faced substantial challenges in tackling continuous optimization problems. Further, in \cite{coccia2022evolution}, the authors revealed a strategic shift from a predominant focus on hardware development towards a more balanced emphasis on software innovation. This transition underscores an evolving landscape in quantum computing, marked by accelerated advancements in quantum image processing, machine learning, and sensory technologies. Addressing the complexities in quantum networking, \cite{chang2022order} focuses on the intricacies of entanglement swapping and the optimization challenges in quantum repeater chains due to their non-associative nature. The exploration of quantum data networks (QDNs) in \cite{qiao2022quantum} presented a distributed computing model that aimed to circumvent the scalability limitations of contemporary quantum computers. By facilitating quantum state information exchange among geographically dispersed devices, this model aspires to aggregate the computational power of multiple quantum-enabled devices into a singular, more potent quantum computing entity. Nevertheless, the implementation of QDNs encounters significant challenges, including the no-cloning theorem, quantum channel loss, the stochastic nature of entanglement link establishment, rapid decoherence of quantum states, and inefficiencies in single-photon detection. The analytical research into transport layer protocols for quantum data networks, discussed in \cite{zhao2023distributed}, assessed their potential to enable the reliable and efficient exchange of quantum data bits (qubits) across a network of quantum computers. The authors highlighted the critical challenge of managing quantum memory, a scarce resource in quantum computing, necessitating sophisticated mechanisms for its dynamic allocation and reallocation. A co-design framework named QuantumFlow, in \cite{jiang2021co}, introduced an innovative method of representing data as unitary matrices. This enables the encoding of inputs into a reduced number of qubits, substantially lowering computational cost complexity compared to classical computing. However, the practical application and scalability of the QuantumFlow framework are contingent upon the current limitations of quantum hardware, including qubit availability and error rates. Lastly, the authors in \cite{ferrari2023modular} addressed the execution challenges of quantum algorithms that surpass the qubit capacities of existing quantum processors. By advocating for the utilization of the quantum Internet to distribute computational tasks across multiple quantum processing units (QPUs), this approach aims to overcome scalability challenges. The efficacy of this framework, however, is intrinsically linked to the development and availability of a robust quantum Internet infrastructure, highlighting a pivotal area for future research and development in the field of quantum computing.

\subsection{Emerging Quantum Frameworks and Protocols}
Shifting to sector-specific applications, \cite{cooper2021exploring} investigated the potential of quantum computing in transportation modeling. However, it confronted the challenge of translating classical algorithms into quantum paradigms due to quantum operations' unitary and reversible nature. In \cite{zhang2021quantified}, authors presented a strategy for optimizing edge server placements, using binary and quantum encoding, but facing challenges due to its computational intensity. In \cite{azad2022solving}, authors applied quantum approximate optimization algorithms to the vehicle routing problem, illustrating the challenges posed by the emerging stage of quantum devices. Innovative approaches and protocols were highlighted in \cite{narottama2023federated} and \cite{ren2022nft}. The authors proposed novel frameworks for wireless communication and IoV, integrating federated learning with quantum teleportation and combining non-fungible tokens with quantum algorithms. However, both the schemes faced practical challenges of quantum neural network complexity and IoV scalability. In \cite{wang2022asynchronous}, authors introduced an asynchronous entanglement distribution protocol for quantum networks. However, the asynchronous nature requires sophisticated coordination among various network nodes to avoid entanglement swapping conflicts, making the protocol more complex than synchronous alternatives. The exploration continues with \cite{pramanik2022optimization} and \cite{leonidas2023qubit}, discussing quantum-classical hybrid methods in sensor placement and vehicle routing, acknowledging technological limitations and complexity. Additionally, in \cite{bhavsar2023classification}, authors applied quantum machine learning to asteroid classification, facing challenges with quantum processors and data management. In contrast, \cite{winker2023quantum} and \cite{wang2021resource} focused on database query processing and IoV resource management, highlighting scalability, reliability, and current quantum computing limitations. The research outlined in \cite{taneja2024quantum} explored the application of reconfigurable intelligent surfaces (RISs) within a quantum computing-based algorithm for RIS selection, aiming to optimize the transmission of qubits from mobile nodes to access points. This methodology seeks to enhance network performance and minimize energy consumption while considering the dynamic mobility of vehicular networks and the consequent variability in line-of-sight (LoS) conditions. Despite the improvements in communication reliability afforded by the use of multiple RISs, the technique needs advanced control algorithms to manage the RISs effectively in a constantly changing vehicular environment.

Distinct from existing literature, we employ entangled quantum feature mapping for intricate data encoding, a method not explored in previous studies. This allows for a more accurate representation of sensor data in quantum states. Additionally, using a parameterized ansatz circuit for quantum state transformation and a fidelity-based cost function for model optimization provides a more precise and adaptable approach than traditional methods. Additionally, applying the quantum Fourier transform and Toffoli gates for state manipulation advances existing literature, offering a novel solution in quantum environmental monitoring.

\section{Problem Statement}
This research addresses explicitly three interconnected challenges: first, the conversion of sensor data from vehicles into a quantum-compatible format; second, the application of a quantum variational classifier for accurate prediction of gas concentrations; and third, the use of quantum entanglement to facilitate rapid, real-time transmission of hazard alerts within a vehicular network.

\subsection{Conversion of Sensor Data to Quantum Data}
To address the first challenge, we focus on the data obtained from the vehicle's sensors, referred to as vector \( \mathbf{S} \). This vector represents the gas concentration. This vector is crucial in the overall quantum computation process as it is the initial data point converted into a quantum state. Our goal is to express these vector measurements as a quantum state \( \ket{\lambda} \) in Hilbert space. The quantum state initialized to 0 undergoes the implementation of the rotation, denoted as \( \mathbb{R_{\theta}} \). The magnitude of angular displacement, represented by \( \theta \), is determined based on the sensor data.
\begin{equation}
\ket{\lambda} = \mathbb{R}_\theta(\mathbf{S})) \ket{\lambda_0}    
\end{equation}
The problem is determining the ideal alignment between classical input and quantum states to ensure the accuracy of subsequent quantum computations.

\subsection{Variational Quantum Classifier for Prediction}
In our second challenge, we design a VQC. This classifier operates through a particular quantum circuit, represented as \( U(\mathbf{\Phi}) \). The \( \mathbf{\Phi} \) is a set of adjustable parameters fine-tuned during the VQC's training phase. The classifier makes these adjustments to make accurate predictions based on the quantum state \( \ket{\lambda} \) derived from the vehicle's sensor data. The problem state is to find probability distribution over potential hazard levels.
\begin{equation}
    P(\text{hazard level} | \mathbf{S}) = \text{Tr}\left[ \Pi \cdot U(\mathbf{\Phi}) \ket{\lambda}\bra{\lambda} U^\dagger(\mathbf{\Phi}) \right]
\end{equation}
Here, \( \Pi \) signifies the projector onto the measurement basis for different hazard levels, with Tr representing the trace operation. The effectiveness of the VQC in accurately predicting hazard levels will be a key focus of our validation efforts.

\subsection{Quantum Entanglement for Real-Time Alert System}
Finally, we address the challenge of rapid, real-time transmission of hazard alerts using quantum entanglement. The entire vehicular network is represented by a multipartite entangled state \( \ket{\Psi} \). The alert system's responsiveness is conceptualized via entanglement-based quantum teleportation, where an alert encoded in a quantum state \( \ket{\lambda} \) is transmitted to a designated vehicle or control center. This process is depicted as:
\begin{equation}
\ket{\lambda}_{\text{target}} = \sigma_x^a \sigma_z^b (I \otimes \bra{\beta_{ab}}) (\ket{\lambda} \otimes \ket{\Psi})    
\end{equation}
In this equation, \( \bra{\beta_{ab}} \) denotes the Bell state measurement, \( \sigma_x \) and \( \sigma_z \) are the Pauli matrices, and \( a, b \) are the classical bits derived from the measurement, which guide the teleportation process. 

\section{System Overview and Architectural Design}

Our proposed architectural design is fundamentally driven by the need to accurately and efficiently detect hazardous gases. Our approach not only governs the manipulation of quantum states but also seamlessly integrates these states into vehicle network communication protocols. 
\begin{figure}[ht]
    \centering
    \includegraphics[width=0.8\columnwidth]{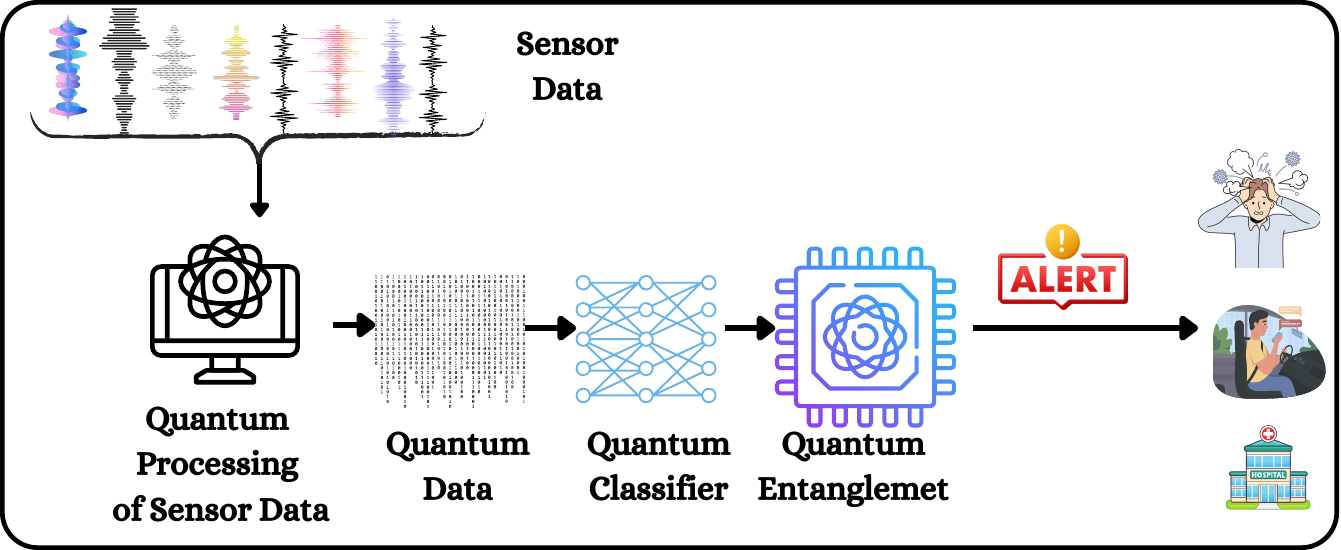}
    \caption{Quantum Process for qIoV Framework}
    \label{qpr}
\end{figure}
Fig. \ref{qpr} presents the model for environmental monitoring. This system harnesses the power of quantum computing to process data collected from various sensors equipped in vehicles. The collected data undergoes normalization and is then encoded into quantum states, enabling compatibility with quantum computational processes. This results in a sophisticated and multifaceted system uniquely built to meet the complex demands of modern environmental monitoring.

\subsection{Data Collection and Quantum State Encoding}
We equipped each vehicle with multiple sensors (Ref. Fig. \ref{sense}) i.e., MQ2, MQ3, MQ5, MQ6, MQ7, MQ8, MQ135 and GPS for air quality evaluation and location monitoring. These sensors form the foundation of the Quantum Mesh Network Fabric's (QMF's) monitoring capabilities. Once the process of data collection is finished, the data is subjected to normalization and quantum encoding. Sensor data is standardized to enable its utilization in quantum computing. The data is then stored in quantum states having multiple simultaneous quantum states \( \ket{\lambda} \). The importance of quantum state encoding resides in its ability to transform conventional sensor data into a format that quantum operations can process. The encoding technique is represented by the function \( E: S \rightarrow \ket{\lambda} \).

\begin{figure}[ht]
    \centering
    \includegraphics[width=95mm]{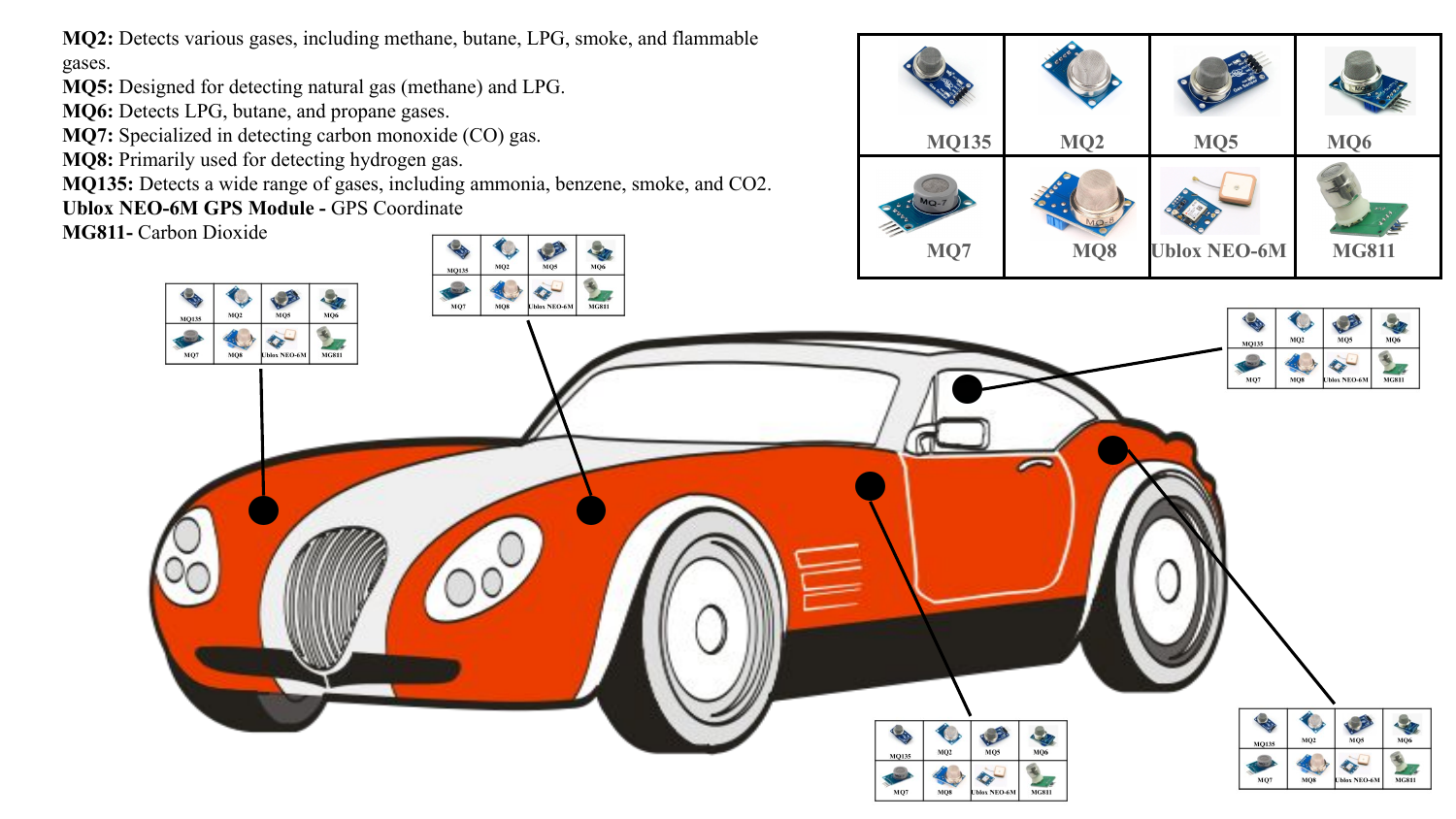}
    \caption{Environmental Sensing using Vehicular Sensors}
    \label{sense}
\end{figure}

\subsection{Quantum Graph and Network State Representation}
The vehicles are considered as quantum entities and present a dynamic quantum graph \( \mathcal{G}(t) \). The edges represent hypothetical quantum entanglements, while the nodes represent vectors of states in Hilbert space $(\mathcal{H})$. The network's state at any given time $t$ is defined by a density matrix \( \rho(t) \) that includes the quantum correlations between the individual states of the vehicles. The Lindblad master equation governs the dynamics of \( \rho(t) \) as it evolves.

\subsection{Communication Model: Quantum Teleportation}

The QMF achieves real-time communication through the utilization of quantum teleportation. During monitoring hazardous gasses, the detected hazards are represented and encoded using a quantum state called \( \ket{\lambda} \). The state in question becomes entangled with the multipartite state of the network, which is represented as  \( \ket{\Psi} \). Bell state measurements are employed in teleportation approach to produce classical bits that aid in reconstructing the target \( \ket{\lambda} \).

\section{Proposed Methodology}
This section provides a detailed explanation of our framework. Our research focuses on three key processes: data conversion, entanglement-based communication, and error management. Initially, multi-dimensional sensor data is transformed into quantum-compatible states using an optimized encoding scheme that precisely calibrates rotation angles to reflect distinct environmental sensor patterns. Subsequently, we deploy a multipartite entanglement strategy, tailored for vehicular networks, to facilitate rapid and reliable environmental alert dissemination, employing quantum teleportation with vehicle-specific entanglement patterns. To counteract the quantum noise and computational errors inherent in dynamic vehicular environments, our approach incorporates real-time adaptive quantum circuit optimization. This step leverages stabilizer codes for error detection and correction, dynamically adjusting quantum gate parameters to preserve the fidelity of quantum communications. Together, these processes establish a robust quantum-enhanced system for environmental monitoring and hazard alert dissemination within vehicular networks.

\subsection{Quantum Data and Encoding and Conversion}

During the initial stage of our process, we collect environmental data from a range of vehicle sensors, specifically MQ2, MQ3, MQ5, MQ6, MQ7, MQ8, and MQ135. Each sensor produces a unique set of environmental data, represented by the notation $x_1, x_2, \ldots, x_7$. Symbolized as \( \mathbf{S} \in \mathbb{R}^{m \times n} \), where $n$ represents the number of sensors and $m$ represents the number of time-dependent data points per sensor. Every occurrence of $x_i$ adds to a thorough dataset, filling a specific row or column in \( \mathbf{S} \), which is then analyzed further.

\subsubsection{Normalization Process}

Given the extensive range of sensor data, standardization is critical. It ensures that the total data collected from each sensor $\mathbf{S_i}$ corresponds precisely to the final quantum state. The data is transformed using Min-Max scalar normalization, which maps the values to a range between $[0,1]$. Eq. \eqref{minmax} is used to calculate the normalized value.
\begin{equation}
s'_{ij} = \frac{s_{ij} - \min(\mathbf{S}_i)}{\max(\mathbf{S}_i) - \min(\mathbf{S}_i)}    
\label{minmax}
\end{equation}
This ensures the data is appropriate for quantum processing by establishing a correlation between each reading and the sensor's range.

\subsubsection{Quantum State Encoding via Entangled Feature Mapping}

Our approach utilizes an entangled quantum feature Map (EQFM) for quantum state encoding (Ref. Fig. \ref{zz}), where entanglement significantly enhances data representation. The encoding, $U_{\theta}^{E}(\boldsymbol{\Phi})$, applies normalized data to a circuit combining controlled-Z and Hadamard gates, producing a highly entangled state suited for complex data patterns. The circuit is defined as:
\begin{equation}
U_{\theta}^{E}(\boldsymbol{\Psi}) = \bigotimes_{i=1}^{7} \left( H_i \cdot CZ_{i,t}(\boldsymbol{\Psi}_i) \right),
\end{equation}

\begin{figure}[ht]
\centering
\includegraphics[width=\linewidth]{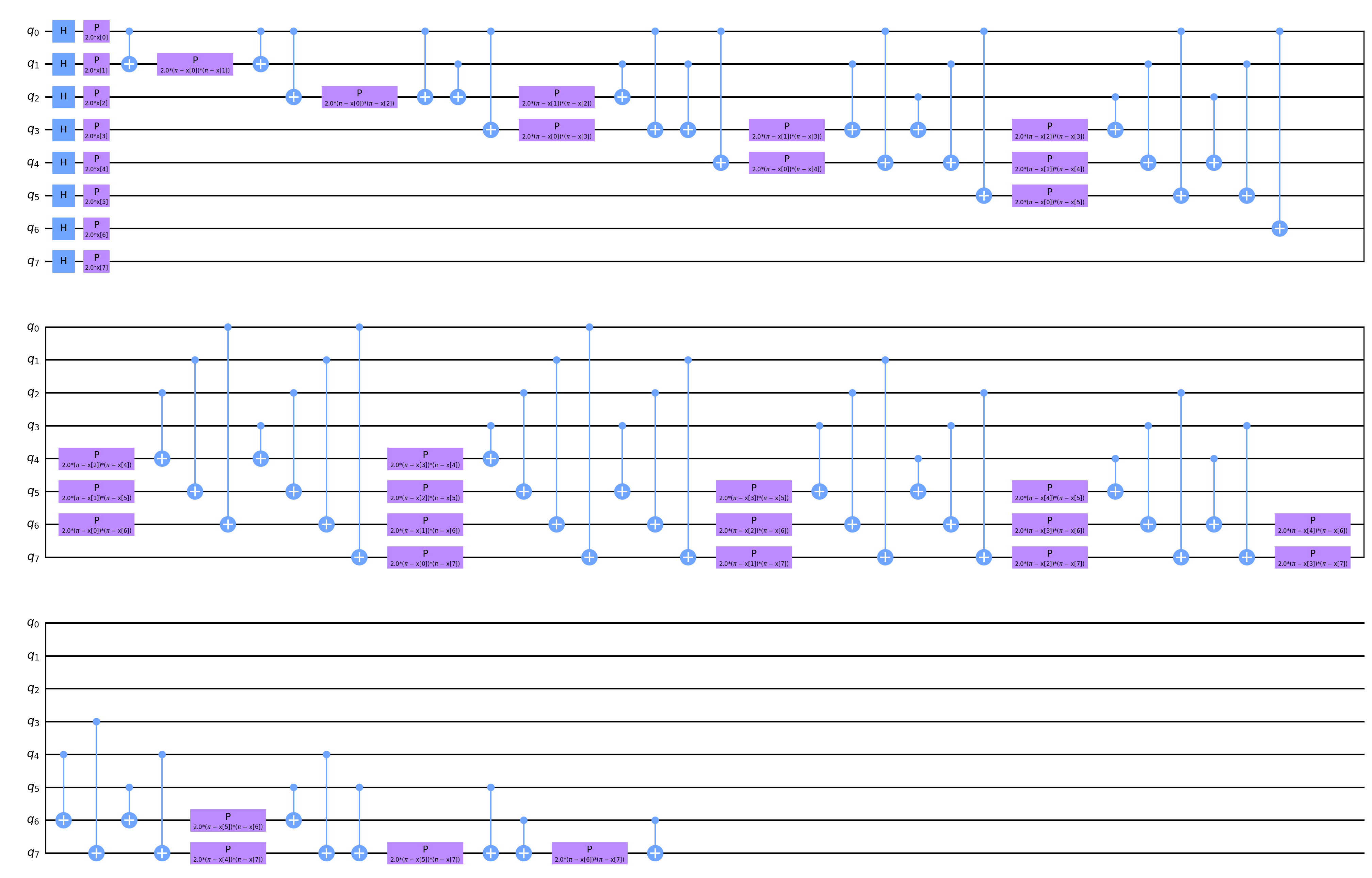}
\caption{Entangled Quantum Feature Map for Advanced Data Encoding}
\label{zz}
\end{figure}

\subsubsection{Parameterized Ansatz for Quantum State Evolution}

We employ a parameterized ansatz circuit $|\omega(\boldsymbol{\theta})\rangle$ (Ref. Fig. \ref{ans}) for quantum state transformation. This includes multi-qubit gates (iSWAPs, CPHASE) and adaptive single-qubit rotations, enhancing the manipulation of quantum states. The processing of evolved sensor data through this ansatz is:

\begin{figure}[ht]
\centering
\includegraphics[width=\linewidth]{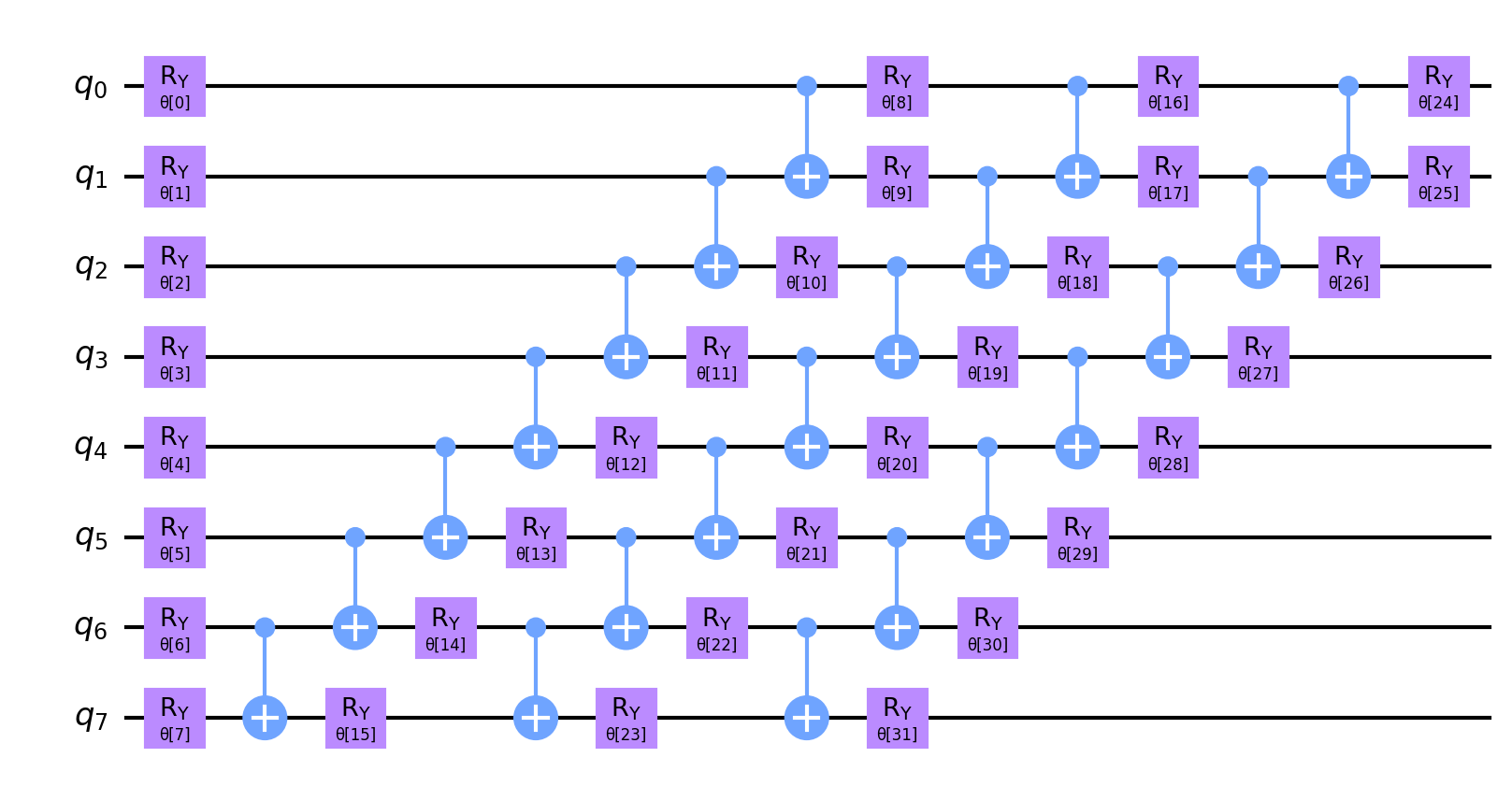}
\caption{Enhanced Ansatz Circuit for Advanced Quantum State Evolution}
\label{ans}
\end{figure}

\begin{equation}
|\omega(\boldsymbol{\theta})\rangle = U_{\text{dynamic}}^{(m)}(\boldsymbol{\theta}_{\text{dynamic}}) \cdot U_{\text{multi-qubit}}^{(m)}(\boldsymbol{\theta}_{\text{multi-qubit}}),
\end{equation}

\subsubsection{Quantum Data Transformation Technique}

Our framework encodes normalized sensor data \( \mathbf{S'} \) into quantum states using a mapping function \( h: \mathbf{S'} \rightarrow \vec{\theta} \), linking data points to rotation angles for optimal quantum circuit utilization. For each qubit \( q_i \), the rotation gate \( R_{\text{advanced}}(\theta_i) \) is:

\begin{equation}
R(\theta_i) = \begin{pmatrix} \cos(\theta_i) & -i\sin(\theta_i) \\ -i\sin(\theta_i) & \cos(\theta_i) \end{pmatrix},
\end{equation}

This advanced gate ensures nuanced data encoding. The mapping function \( h \) is adaptive, optimizing sensor data representation in the quantum domain.

\subsubsection{Quantum Encoding and State Preparation}

Sensor data encoding into a quantum state \( U(\mathbf{S'}) \) involves both single and multi-qubit rotations. The resulting quantum state \( |\lambda(\boldsymbol{\theta})\rangle \) is crucial for quantum computations. The encoding is performed as follows:

\begin{equation}
U(\mathbf{S'}) = \bigotimes_{i=1}^{n} R(\theta_i(s'_i))
\end{equation}

The initial state \( (\ket{\lambda}_0) \) evolves into the encoded state \( \ket{\lambda} \), with precision in \(\theta\) calculations essential for state fidelity.

\subsection{Development of Quantum Variational Classifier for Toxic Gas Prediction}

Our primary research aims to develop a quantum classifier that can reliably forecast quantities of dangerous gasses. This section thoroughly explains the internal mechanisms of the VQC, beginning from its creation and continuing to its practical implementation. The development of the VQC involves a multi-stage approach. Initially, the data obtained from traditional sensors is converted into a format that quantum computers can understand. We evaluate the correctness of our quantum model by employing a technique called fidelity-based cost function. In the final phase, we utilize an advanced optimization technique to enhance the performance of our model.

\subsection{Algorithmic Overview}
The construction of the VQC, as depicted in Algorithm \ref{alg:vqc_gas_prediction_advanced}, entails several fundamental steps:
\begin{algorithm}[ht]
\caption{Advanced Variational Quantum Classifier with Adaptive Optimization}
\label{alg:vqc_gas_prediction_advanced}
\begin{algorithmic}[1]
\STATE \textbf{Input:} Sensor data $\{x^{(i)} = (x_1^{(i)}, x_2^{(i)}, \ldots, x_7^{(i)})\}_{i=1}^{N}$, Variational parameters $\boldsymbol{\theta}$, Number of layers $M$
\STATE \textbf{Output:} Predicted gas types for given sensor data
\FOR{$i=1$ to $N$}
    \STATE Encode sensor data into quantum states: $|\lambda^{(i)}(\boldsymbol{\theta})\rangle \leftarrow U_{\text{entangle}}^{(M)}(\boldsymbol{\theta}_{\text{entangle}}) \cdot U_{\text{rotation}}^{(M)}(\boldsymbol{\theta}_{\text{rotation}}) |0\rangle^{\otimes n}$
\ENDFOR
\STATE Initialize cost function: $C(\boldsymbol{\theta}) \leftarrow 0$
\FOR{$i=1$ to $N$}
    \STATE Execute quantum circuit: $U(\boldsymbol{\theta})|\lambda^{(i)}(\boldsymbol{\theta})\rangle$
    \STATE Measure probabilities $P(y_j|\boldsymbol{\theta})$ for gas classes $y_j$
    \STATE Calculate fidelity $F(|\lambda_{\text{expected}}^{(i)}\rangle, U(\boldsymbol{\theta})|\lambda^{(i)}(\boldsymbol{\theta})\rangle)$
    \STATE Update cost function: $C(\boldsymbol{\theta}) \leftarrow C(\boldsymbol{\theta}) + (1 - F(|\lambda_{\text{expected}}^{(i)}\rangle, U(\boldsymbol{\theta})|\lambda^{(i)}(\boldsymbol{\theta})\rangle))^2$
\ENDFOR
\STATE \textbf{Optimization:}
\STATE Set initial learning rate $\eta$, threshold $\epsilon$, and maximum iterations $K$
\STATE Initialize $\boldsymbol{\theta}_0$ randomly
\STATE $k \leftarrow 0$
\WHILE{$k < K$ or $\Delta C(\boldsymbol{\theta}_k) > \epsilon$}
    \STATE Calculate gradient $\nabla C(\boldsymbol{\theta}_k)$
    \STATE Adjust learning rate: $\eta_k \leftarrow \eta / (\sqrt{k+1})$
    \STATE Update parameters: $\boldsymbol{\theta}_{k+1} \leftarrow \boldsymbol{\theta}_k - \eta_k \nabla C(\boldsymbol{\theta}_k)$
    \STATE $k \leftarrow k + 1$
\ENDWHILE
\RETURN Predicted gas types based on optimized $\boldsymbol{\theta}$
\end{algorithmic}
\end{algorithm}

\subsubsection{Quantum State Encoding}

Given classical sensor data $x^{(i)} = (x_1^{(i)}, x_2^{(i)}, \ldots, x_7^{(i)})$ for $N$ instances, we prepare the quantum state, denoted as $|\lambda^{(i)}(\boldsymbol{\theta})\rangle$, by using specially designed quantum circuits. These circuits are adjustable ('parameterized') and tailored to represent our sensor data in the quantum realm.
\begin{equation}
|\lambda^{(i)}(\boldsymbol{\theta})\rangle \leftarrow U_{\text{entangle}}^{(M)}(\boldsymbol{\theta}_{\text{entangle}}) \cdot U_{\text{rotation}}^{(M)}(\boldsymbol{\theta}_{\text{rotation}}) |0\rangle^{\otimes n}.
\end{equation}
Here, $U_{\text{entangle}}^{(M)}(\boldsymbol{\theta}_{\text{entangle}})$ and $U_{\text{rotation}}^{(M)}(\boldsymbol{\theta}_{\text{rotation}})$ represent entangling and single-qubit rotation gates with $M$ layers, parameterized by $\boldsymbol{\theta}_{\text{entangle}}$ and $\boldsymbol{\theta}_{\text{rotation}}$ respectively.

\subsubsection{Cost Function Formulation}

A 'fidelity-based cost function' in quantum computing measures how closely our quantum model's predictions match the expected outcomes. It serves as a scoring system that tells us how accurate our model is. The fidelity-based cost function $C(\boldsymbol{\theta})$ is defined to capture the discrepancy between expected and measured quantum states:
\begin{equation}
C(\boldsymbol{\theta}) \leftarrow \sum_{i=1}^{N} \left(1 - F(|\lambda_{\text{expected}}^{(i)}\rangle, U(\boldsymbol{\theta})|\lambda^{(i)}(\boldsymbol{\theta})\rangle)^2\right).
\end{equation}
Here, $F(|\lambda_{\text{expected}}^{(i)}\rangle, U(\boldsymbol{\theta})|\lambda^{(i)}(\boldsymbol{\theta})\rangle)$ represents the fidelity between the expected and measured quantum states.

\subsubsection{Adaptive Optimization Scheme}

Our optimization scheme uses a step-by-step ('iterative') approach to improve our model gradually. In each step, we adjust the learning rate $\eta_k$ (a measure of how much the model changes in each step) to ensure steady and effective improvement.
\begin{equation}
\boldsymbol{\theta}_{k+1} \leftarrow \boldsymbol{\theta}_k - \eta_k \nabla C(\boldsymbol{\theta}_k),
\end{equation}
where $\nabla C(\boldsymbol{\theta}_k)$ denotes the gradient of the cost function concerning the variational parameters. The quantum state encoding, fidelity-based cost function, and adaptive optimization scheme synergistically create a robust framework for fine-tuning variational parameters based on fidelity alterations. This complexity is fundamental in capturing subtle variations within quantum states derived from sensor data, facilitating precise gas-type predictions leveraging quantum computational advantages.

\subsubsection{Extensive Quantum State Processing and Measurement for Hazard Level Assessment}

After the training phase, the VQC deploys its fine-tuned circuit on freshly encoded quantum sensor data, represented as \( \ket{\lambda} \). The resulting state, \( \ket{\lambda'} \), achieved through the operation \( U(\mathbf{\Phi}) \ket{\lambda} \), reflects the VQC's analysis of various toxic gas concentrations. To extract actionable insights, the state \( \ket{\lambda'} \) undergoes measurement using a series of projective operators \( \{\Pi_i\} \), each linked to a specific level of hazard. The likelihood of encountering a particular hazard level, based on the sensor data, is determined using the formula:
\begin{equation}
P(\text{hazard level}_i | \mathbf{S}) = \text{Tr}\left[ \Pi_i \cdot U(\mathbf{\Phi}) \ket{\lambda}\bra{\lambda} U^\dagger(\mathbf{\Phi}) \right].    
\end{equation}

This method, which weaves together data encoding processes, quantum circuit construction, training methodologies, and practical execution, thoroughly delineates the development of a VQC for predicting toxic gas levels. It underscores the system's potential to significantly boost public safety by providing early warnings of hazardous conditions.

\subsection{Quantum Entanglement Process for Toxic Gas Detection}

In contemporary quantum physics research, the phenomenon of quantum entanglement stands as a foundation, particularly in the application to QMF. Our qIoV framework is designed to facilitate communication within vehicular networks and provide critical monitoring of hazardous gasses. Entanglement, in this context, refers to a distinct quantum behavior where the states of particles (i.e., electrons or photons) become interlinked so that the condition of one particle instantaneously affects its counterpart, transcending the constraints of spatial distance.

\begin{figure}[ht]
    \centering
    \includegraphics[width=0.9\linewidth]{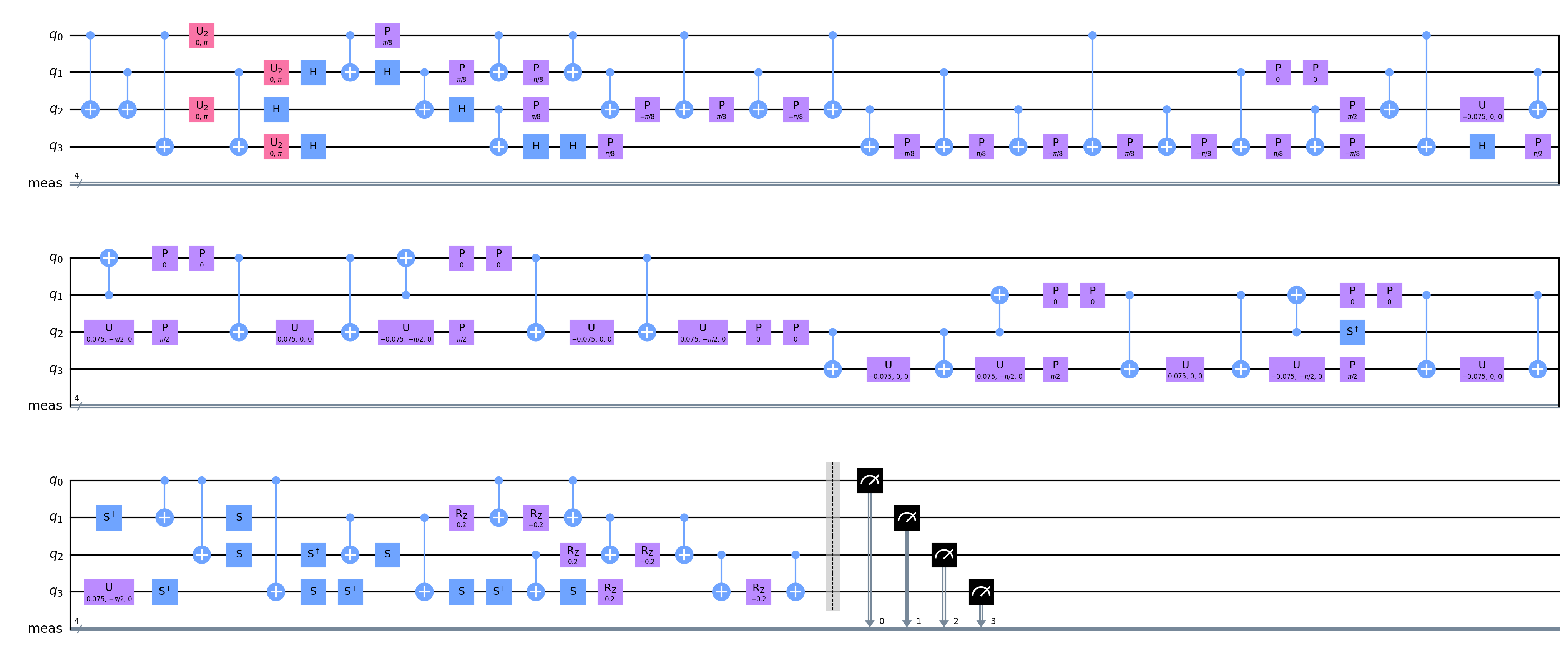}
    \caption{Quantum Circuit Representation of Communication}
    \label{fig:quantum_circuit}
\end{figure}

Our model leverages quantum entanglement and teleportation for efficient, secure state transitions, essential in real-time environmental data transmission. It is based on a quantum circuit with a three-qubit quantum register, optimized through an analysis of vehicular dynamics and illustrated in Fig. \ref{fig:quantum_circuit}.

We use Bell State Preparation for qubit entanglement, assigning each vehicle a quantum identifier, $i$, and a corresponding server qubit, $j$, in a network of $n$ vehicles.

\begin{equation}
|\Psi\rangle_i \xrightarrow{\text{Bell}} \frac{1}{\sqrt{2}}(|00\rangle + |11\rangle)_{i, n+j}
\end{equation}

Qubits initially transition to a Bell state, a more advanced technique than Hadamard gates, for superior information encoding.

\begin{equation}
|\Psi\rangle_i \xrightarrow{\text{Bell}} \frac{1}{\sqrt{2}}(|00\rangle + |11\rangle)_{i, j}
\end{equation}

Next, quantum Fourier transform (QFT) manipulates entanglement states between qubits for data embedding and transfer.

\begin{equation}
|\Psi\rangle_i, |\Psi\rangle_{i+1} \xrightarrow{\text{QFT}} QFT(|\Psi\rangle_i), QFT(|\Psi\rangle_{i+1})
\end{equation}

Enhancements include a Toffoli gate for strengthening quantum correlations and quantum phase estimation (QPE) for precise phase calculations, ensuring dynamic network interactions. The protocol concludes with Quantum Teleportation and a controlled-$Z$ rotation gate for advanced quantum manipulations.

\begin{algorithm}[h]
\caption{Quantum Communication Protocol}
\label{quantum_protocol}
\begin{algorithmic}[1]
\REQUIRE vehicle count $(n)$, number of rounds $(R)$, quantum register size per vehicle $(m)$
\ENSURE Quantum communication protocol
\STATE Initialize a quantum circuit with $3m = 3n$ qubits.
\FOR{$r_i$ = $1$ to $R$}
    \FOR{each vehicle $v_i$, $i = 1$ to $n$}
        \STATE Prepare Bell state for qubits $(3i-2, 3i-1)$: $|\Psi\rangle_{3i-2, 3i-1} \xrightarrow{\text{Bell}} \frac{1}{\sqrt{2}}(|00\rangle + |11\rangle)$
        \STATE Apply QFT on the third qubit of each vehicle: $|\Psi\rangle_{3i} \xrightarrow{\text{QFT}} QFT(|\Psi\rangle_{3i})$
    \ENDFOR
    \FOR{each pair of vehicles $(v_i, v_j)$, $i \neq j$}
        \STATE Apply Toffoli gate to strengthen quantum correlations: $\text{Toffoli}(|\Psi\rangle_{3i-1}, |\Psi\rangle_{3j-1}, |\Psi\rangle_{3i})$
        \STATE Apply QPE for precise phase calculations: $\text{QPE}(|\Psi\rangle_{3i}, |\Psi\rangle_{3j})$
    \ENDFOR
    \FOR{each vehicle $v_i$}
        \STATE Perform Quantum Teleportation using a controlled-$Z$ rotation gate: $\text{Teleport}(|\Psi\rangle_{3i-2}, |\Psi\rangle_{3i-1}, |\Psi\rangle_{3i})$
    \ENDFOR
\ENDFOR
\STATE Measure all qubits.
\STATE Transpile the circuit for optimization.
\RETURN Quantum circuit representing the communication protocol.
\end{algorithmic}
\end{algorithm}

\section{PERFORMANCE EVALUATION}
This section assesses the qIoV framework, focusing on its proficiency in detecting toxic gases with enhanced precision. We concentrated on accurately identifying different gas types by utilizing vehicle-mounted sensor data. The pair plot method was pivotal in analyzing data correlations and enhancing testing accuracy. Additionally, we probed quantum data distributions under varying conditions on the IBM OpenQSAM 3 platform, which was equipped with a 127 Qubit system.

\subsection{Implementation Details of VQC on the IBM OpenQSAM 3 Platform}
The experiments were carried out by designing a parameterized quantum circuit specific to the VQC. Utilizing the quantum circuit composer provided by the OpenQSAM 3 platform, we arranged Hadamard, Pauli-X, Controlled-Z, and rotation gates to create an initial entangled state that could capture the complexity of the sensor data. Sensor data encoding was performed using a customized quantum feature map, which was developed to transduce the classical gas concentration readings into quantum states. The optimization was executed through OpenQSAM 3’s hybrid quantum-classical interface. At each iteration, the platform computed the cost function's value by executing the quantum circuit with the current parameters and measuring the output. The fidelity-based cost function was then evaluated by calculating the overlap between the VQC's output state and the ideal state corresponding to the correct classification of gas types. Parameter updates were governed by a gradient-based optimization algorithm. The gradients were estimated using the parameter-shift rule. To account for the effects of quantum noise and potential hardware imperfections, we introduced depolarizing noise into our simulations. This allowed us to assess the VQC’s resilience and adjust our circuit design and error mitigation strategies accordingly. Post-optimization, the trained VQC was deployed within the OpenQSAM 3 environment for the classification of new sensor data. The real-time execution of the quantum circuit allowed for the immediate classification of gas concentrations, and the results were integrated into the IoV system to trigger appropriate hazard alerts. The experimental results, including precision, recall, and F1-Score metrics, were benchmarked against classical machine learning classifiers to validate the quantum advantage offered by the VQC. Further, the exploration of sensor data was based on an advanced pair plot visualization technique. This approach unraveled complex correlations and distinct attributes of various gases. The data, categorized into four groups - 'No Gas,' 'Perfume,' 'CO2 and CO Mixture', and 'Toxic Gases Mixture,' was intricately displayed through a matrix of scatter plots and histograms. This visualization provided an understanding of gas detection. Our theoretical analysis also encompassed rotation validation, teleportation protocol efficiency, and quantum entanglement fidelity, offering a holistic view of the system's functionality.

\subsection{Pair Plot Analysis}
The pair plot analysis is used to reveal correlations, outliers, and clusters within the data. In our experimetal analysis, correlations between sensor readings indicate similar sensitivity profiles to certain gas concentrations, while outliers indicates anomalous detections or sensor malfunctions. Clusters of data points, as color-coded by gas type, provide visual evidence of the sensors' discrimination capabilities between different gas concentrations. The data matrix \textbf{S} serves as the foundation for a multidimensional analysis of sensor outputs. The matrix, denoted as \textbf{D}, represents a composite distribution of sensor readings across multiple gas types, enabling a detailed pairwise comparison of sensor responses. Each scatter plot within the matrix contrasts two different sensors' readings, with the data points color-coded to represent distinct gas concentrations, thus allowing for the immediate visual discrimination of distribution patterns. The diagonal of the pair plot matrix provides histogram representations of the sensor readings, offering insights into the univariate distribution for each sensor. Fig. \ref{fig:3}  visually represents the procedure and illustrates how the data from the sensors correspond to the concentrations of the different gasses.

\begin{figure}[ht]
    \centering
    \includegraphics[width=\textwidth,height=130mm]{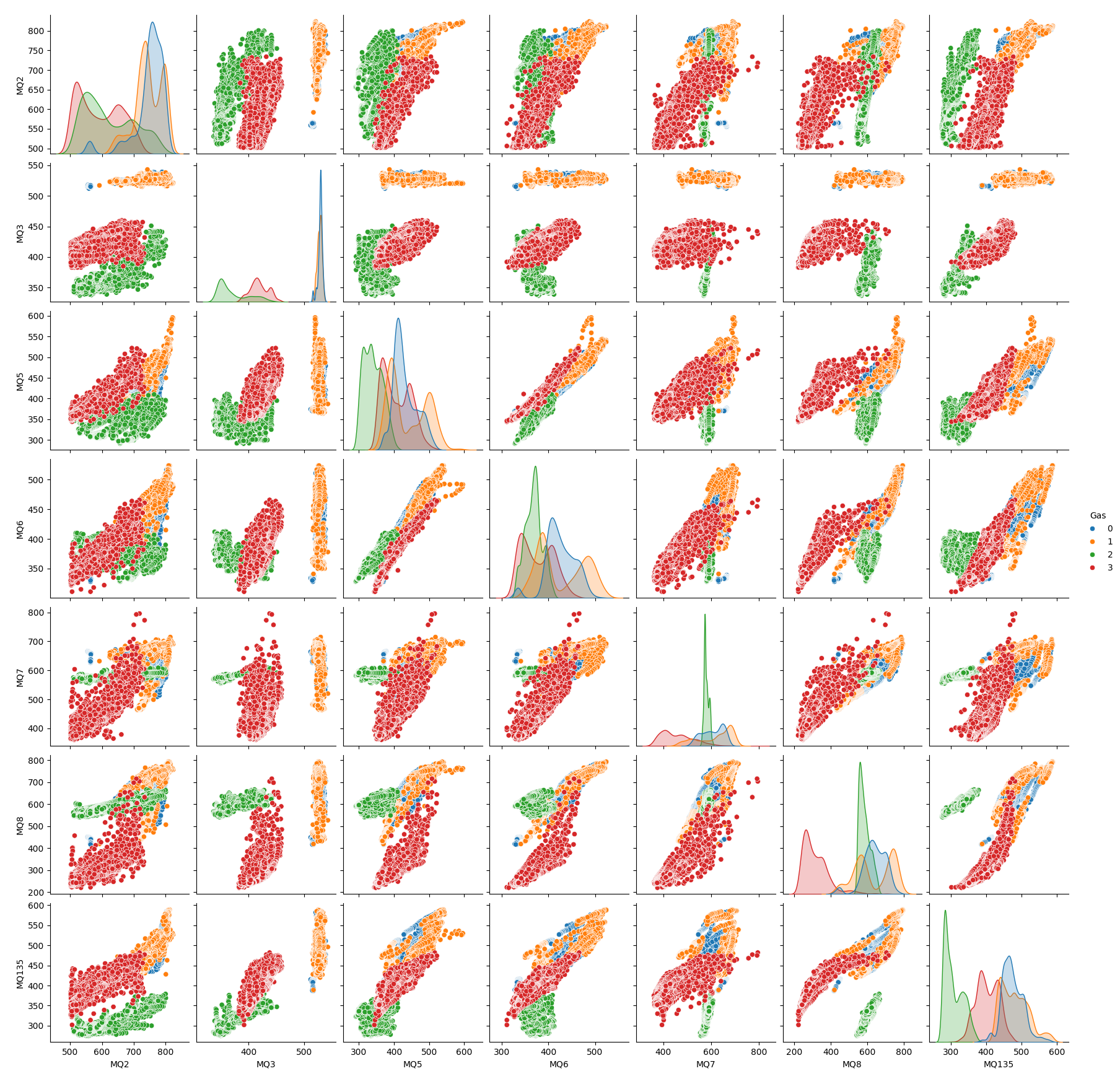}
    \caption{Pair plot showing the relationships between sensors for different gas types.}
    \label{fig:3}
\end{figure}

We conducted thorough sensitivity and reliability testing to assess the sensors' capacity to detect gas concentrations based on our own dataset \footnote{https://github.com/TakMashhido/Gas-Sensors-Measurements-Dataset}. The experimental design consisted of 6,400 data points, categorized into four groups depending on gas concentration. We allocated 1600 data points evenly across all groups to provide a comprehensive and unbiased study. The methodical arrangement of the data was crucial in mitigating biases and enabling a comprehensive evaluation of the sensor's performance in various configurations. In our experiment, positive correlations are evidenced by the clustering of data points along a line in the scatter plots. For instance, MQ2 and MQ7 display a tight linear cluster suggesting a strong positive correlation, indicating that these sensors have similar responses to the varying gas concentrations within the tested range. The degree of spread in the data points reflects the variability in sensor responses. A wide dispersion represents a broad detection range or a sensitivity to a variety of gas types, while a narrow dispersion indicates specificity. Conversely, the presence of distinct clusters along the scatter plot axes indicates a differential response from the sensors to specific gases. This is pivotal, as it demonstrates the array's capability to differentiate between gases, which is essential for accurate environmental monitoring. Further, the experiments reveal the distribution of readings across all tested gas concentrations for a single sensor, which is critical for understanding the sensor's sensitivity range and the likelihood of its activation under different gas concentration levels. Sensors MQ7 and MQ8 have a broader response distribution, indicating versatility in detecting a range of gas concentrations, while MQ3 and MQ6 show a narrower distribution, hinting at more specific gas detection capabilities. Further, the range of values—minimum to maximum for each sensor also indicate their operational range and detection thresholds. 
\begin{itemize}
    \item MQ2: Min = 502.0, Max = 824.0
    \item MQ3: Min = 337.0, Max = 543.0
    \item MQ5: Min = 291.0, Max = 596.0
    \item MQ6: Min = 311.0, Max = 524.0
    \item MQ7: Min = 361.0, Max = 796.0
    \item MQ8: Min = 220.0, Max = 794.0
    \item MQ135: Min = 275.0, Max = 589.0
\end{itemize}
For instance, MQ2's operational range is from 502 to 824, indicating its capability to detect higher concentrations, which is suitable for environments where such conditions are met.

\subsection{Probability Distribution: Theoretical and Real-Time}

In the ideal theoretical probability distribution (Ref. Table \ref{tab1} ), we observe a high probability for the '0000' state across all three bases (Z, X, Y). This indicates that the quantum state encoding and preparation processes are highly effective. The precision in the initial state $|\lambda(\boldsymbol{\theta})\rangle$ evolving into the encoded state $|\lambda\rangle$ with accurate $\boldsymbol{\theta}$ calculations is crucial here. The high fidelity of this state suggests that the quantum entanglement process and the entangled quantum feature map (EQFM) are functioning optimally, creating a highly entangled state that accurately represents the sensor data. The low probability of other outcomes further indicates that the quantum state preparation and encoding minimize errors and noise in the system.

\begin{table}[ht]
    \centering
    \caption{Theoretical Probability Distribution (Ideal)}
    \resizebox{0.5\columnwidth}{!}{
    \begin{tabular}{|c|c|}
        \hline
        \textbf{Basis} & \textbf{Theoretical Probability Distribution} \\
        \hline
        Z-Basis (P\_Z) & $P_Z('0000') = \text{High Probability}$ \\
                       & $P_Z(\text{other outcomes}) = \text{Low Probability}$ \\
        \hline
        X-Basis (P\_X) & $P_X('0000') = \text{High Probability}$ \\
                       & $P_X(\text{other outcomes}) = \text{Low Probability}$ \\
        \hline
        Y-Basis (P\_Y) & $P_Y('0000') = \text{High Probability}$ \\
                       & $P_Y(\text{other outcomes}) = \text{Low Probability}$ \\
        \hline
    \end{tabular}
    }
    \label{tab1}
\end{table}

When introducing errors (Ref. Table \ref{tab2}), the probability of observing the '0000' state decreases, while the probabilities of other outcomes increase across all bases. This change reflects the impact of environmental or systemic errors on the quantum state. The deviation from the ideal distribution is attributed to factors like quantum decoherence and operational errors in the quantum circuit, including gate errors or state preparation and measurement (SPAM) errors. The presence of errors in quantum state evolution, as handled by parameterized ansatz circuit $|\omega(\boldsymbol{\theta})\rangle$, also contributes to these discrepancies.

\begin{table}[ht]
    \centering
    \caption{Theoretical Probability Distribution (with error)}
    \resizebox{0.6\columnwidth}{!}{
    \begin{tabular}{|c|c|}
        \hline
        \textbf{Basis} & \textbf{Theoretical Probability Distribution} \\
        \hline
        Z-Basis Measurement (P\_Z\_Error) & $P_Z\_Error('0000') < \text{High Probability}$ \\
                                          & $P_Z\_Error(\text{other outcomes}) > \text{Low Probability}$ \\
        \hline
        X-Basis Measurement (P\_X\_Error) & $P_X\_Error('0000') < \text{High Probability}$ \\
                                          & $P_X\_Error(\text{other outcomes}) > \text{Low Probability}$ \\
        \hline
        Y-Basis Measurement (P\_Y\_Error) & $P_Y\_Error('0000') < \text{High Probability}$ \\
                                          & $P_Y\_Error(\text{other outcomes}) > \text{Low Probability}$ \\
        \hline
    \end{tabular}
    }
    \label{tab2}
\end{table}

The simulation results for real-time distribution in Table \ref{tab3} show a distribution of probabilities across various states in the Z, X, and Y bases. Unlike the ideal and error-included theoretical distributions, these results demonstrate a more uniform distribution, indicating a dynamic interaction of the quantum system with real-time environmental data. This reflects the adaptive optimization scheme employed in the methodology, where the quantum state is iteratively refined to represent the sensor data better. The variations in probabilities across different states suggest that the quantum system effectively captures the nuances of the environmental data. Higher probabilities for specific states (like '1011', '0010', '1101', etc.) in different bases indicate the successful encoding of sensor data into quantum states and the effective transformation of these states via the parameterized ansatz.

\begin{table}[ht]
    \centering
    \caption{Simulation Results in the Context of Quantum-Based Real-Time Vehicular Communication}
    \resizebox{0.6\columnwidth}{!}{
    \begin{tabular}{|c|c|c|c|||c|c|c|c|}
        \hline
        \textbf{Basis} & \textbf{Z-Basis} & \textbf{X-Basis} & \textbf{Y-Basis} & \textbf{Basis} & \textbf{Z-Basis} & \textbf{X-Basis} & \textbf{Y-Basis} \\
        \hline
        '1011' & 63 & 67 & 62 & '1000' & 71 & 75 & 53 \\
        '0010' & 59 & 63 & 66 & '1110' & 79 & 63 & 66 \\
        '0011' & 57 & 58 & 59 & '0101' & 51 & 62 & 64 \\
        '1101' & 64 & 63 & 70 & '1010' & 75 & 62 & 51 \\
        '0001' & 66 & 59 & 60 & '1100' & 62 & 66 & 70 \\
        '0000' & 56 & 61 & 56 & '1001' & 69 & 64 & 67 \\
        '0111' & 60 & 72 & 74 & '1111' & 71 & 61 & 78 \\
        '0110' & 58 & 67 & 53 & '0100' & 63 & 61 & 75 \\
        \hline
    \end{tabular}
    }
    \label{tab3}
\end{table}

The introduction of depolarizing error in real-time communication further shifts in probabilities. Depolarizing error is a type of quantum noise that leads to the loss of information about the state, causing it to trend toward a maximally mixed state. The observed probabilities in Table \ref{tab4} reflect the resilience of the quantum system against such noise. Although different from the ideal scenario, the system's ability to maintain a coherent distribution of probabilities indicates the robustness of the quantum encoding, state preparation, and adaptive optimization process. These results also underscore the impact of real-world factors and noise on quantum systems and the need for sophisticated error correction and noise resilience techniques in quantum computing.

\begin{table}[ht]
    \centering
    \caption{Results with Depolarizing Error in Real-Time Vehicular Communication}
    \resizebox{0.4\columnwidth}{!}{
    \begin{tabular}{|c|c||c|c|}
        \hline
        \textbf{State} & \textbf{Probabilities} & \textbf{State} & \textbf{Probabilities} \\
        \hline
        '1010' & 71 & '0100' & 51 \\
        '1111' & 76 & '0001' & 74 \\
        '0110' & 78 & '1101' & 67 \\
        '1100' & 59 & '1110' & 59 \\
        '1011' & 64 & '1001' & 60 \\
        '0010' & 61 & '0000' & 55 \\
        '0101' & 73 & '0011' & 66 \\
        '0111' & 57 & '1000' & 53 \\       
        \hline
    \end{tabular}
    }
    \label{tab4}
\end{table}

\subsection{VQC Model Prediction Result}
\subsubsection{Objective Function} Fig. \ref{ofv} illustrates a declining trend in objective function value against iteration number. It indicates the effectiveness of adaptive optimization scheme implemented in our research. The results reflects the iterative refinement of the variational parameters $\theta$ through the application of a gradient descent algorithm paired with an adaptive learning rate. 

\begin{figure}[ht]
\centering
\includegraphics[width=0.5\linewidth]{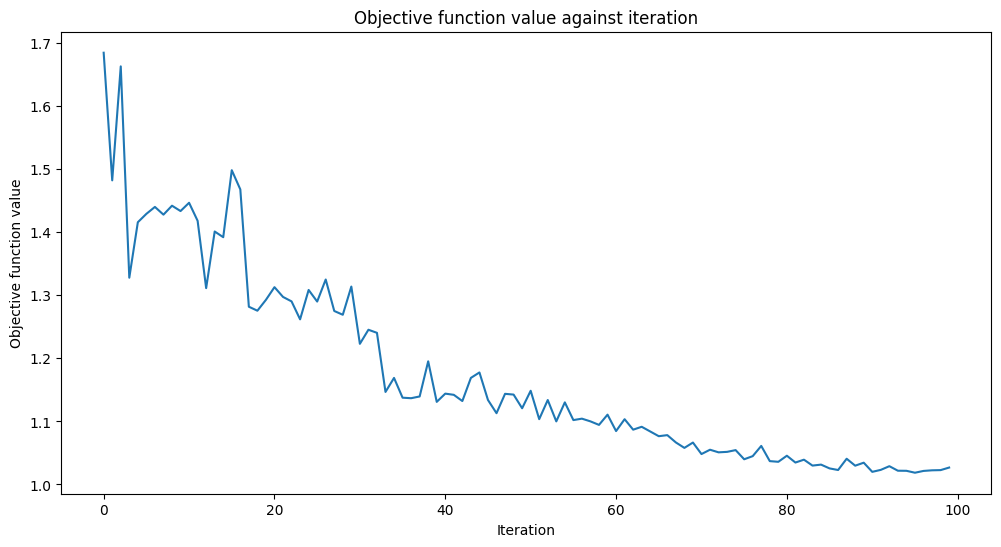}
\caption{Objective Function Values Against Iteration}
\label{ofv}
\end{figure}

The initial fluctuations observed in the graph represent the exploratory phase of the optimization process, where the algorithm searches the parameter space for a local minimum of the cost function. These fluctuations are an anticipated aspect of the optimization, attributable to the adaptive adjustment of the learning rate, which is modulated in response to the gradient history. This exploratory phase is critical for escaping any local minima that are not optimal. As the iterations proceed, the objective function value based on fidelity-based cost function, steadily decreases. This signifies that the optimization technique is effectively honing in on the optimal set of parameters that minimize the cost function. The fidelity-based cost function quantifies the closeness of the predicted quantum state to the true state, with lower values indicating higher fidelity. Therefore, the observed minimization of this function is indicative of the quantum classifier's improved accuracy in predicting gas concentrations. The adaptive learning rate plays a pivotal role in this optimization process. By adjusting the learning rate based on the gradient's magnitude, the algorithm ensures that the parameter updates are neither too large—potentially overshooting the minimum—nor too small, which could result in a prolonged convergence time. This balance accelerates the convergence of the optimization algorithm while maintaining the precision of the parameter updates.

\subsubsection{Confusion Matrix} The confusion matrix in Fig. \ref{conf} shows a strong diagonal, indicating highly accurate positive values for each class. Our quantum state encoding via entangled feature mapping influences this result. The entangled feature map increases the dimensionality of the feature space, which leads to better separation of the classes in the quantum state space. Moreover, the parameterized ansatz for quantum state evolution allows for a richer exploration of the quantum state space, potentially providing a more nuanced separation between classes that appear similar in the classical feature space.

\begin{figure}[ht]
\centering
\includegraphics[width=0.4\linewidth]{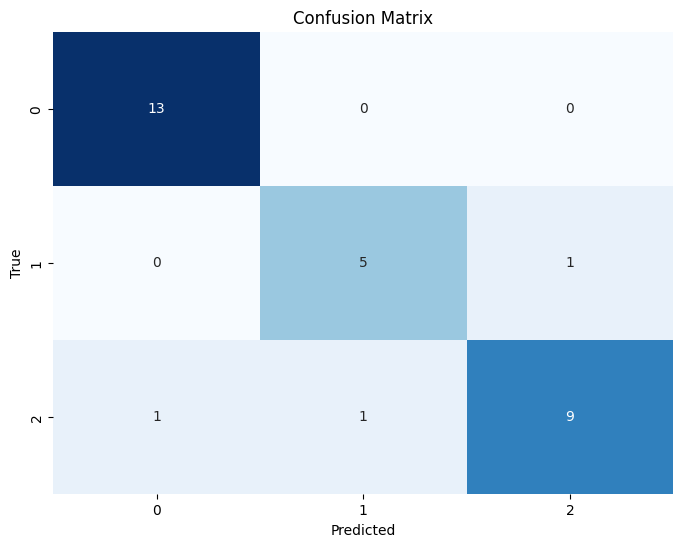}
\caption{Confusion Matrix}
\label{conf}
\end{figure}

\subsubsection{Classification Metrics}
We have categorized the classification metrics results into three classes based on precision and recall as shown in Table \ref{tab5}:

\textbf{Class 0:} Exhibits high precision and recall due to the effective normalization and quantum state encoding methods that preserve the fidelity of the original sensor data, allowing for distinct features of this class to be well-represented in the quantum feature space. This indicates that the quantum circuit has been well-optimized for these features.

\textbf{Class 1:} The balanced precision and recall suggest that while the feature space for this class is well-mapped, there is some overlap with other classes. The EQFM and the ansatz need further refinement to distinguish the features of this class better.

\textbf{Class 2:} The high precision but slightly lower recall implies that when the classifier predicts an instance to be in class '2', it is likely correct, but it misses some actual instances of class '2'. This is due to a slight underrepresentation of class '2' features in the entangled quantum feature map. In the quantum variational classifier, the fidelity-based cost function is minimized, indicating that the quantum states prepared by the circuit are close to the expected states corresponding to each class. The adaptive optimization technique enhances the classifier’s performance by efficiently finding the optimal parameters resulting in the respective classes' highest fidelity states.

\begin{table}[ht]
    \centering
\caption{Classification Performance Metrics}
\resizebox{0.4\columnwidth}{!}{
\begin{tabular}{|c|c|c|c|c|}
\hline
Class & Precision & Recall & F1-Score & Support \\
\hline
0 & 0.93 & 1.00 & 0.96 & 13 \\
1 & 0.83 & 0.83 & 0.83 & 6 \\
2 & 0.90 & 0.82 & 0.86 & 11 \\
\hline
\end{tabular}
}
\label{tab5}
\end{table}

\subsection{Detection Time Analysis}
The gas detection time analysis in Fig. \ref{F1} reveals significant insights into the performance of various detection models. Among the evaluated models, the qIoV Framework consistently outperforms traditional methods such as Support Vector Machine (SVM), Random Forest, and Logistic Regression in terms of detection speed. The dataset, comprising 100 data points, demonstrates that the qIoV Framework achieves detection times significantly lower than its counterparts. The technical superiority of the qIoV Framework is attributed to its optimized algorithmic structure, which leverages advanced data processing and analysis techniques. Our quantum circuits leverage qubits that can exist in multiple states simultaneously (quantum superposition), and they utilize quantum entanglement. These principles allow quantum algorithms to process complex datasets more efficiently than classical algorithms, leading to significant improvements in the speed of gas detection and classification. This efficiency is due to the quantum circuits' ability to perform multiple calculations at once and to quickly navigate the solution space of possible gas detection signatures, identifying hazards more rapidly than traditional computing methods. Unlike the SVM and Random Forest models, which rely on complex decision boundaries and ensemble learning respectively, the qIoV Framework utilizes a more efficient approach to achieve quicker response times. This efficiency is crucial for real-time gas detection applications where timely response can mitigate risks and prevent hazardous incidents.

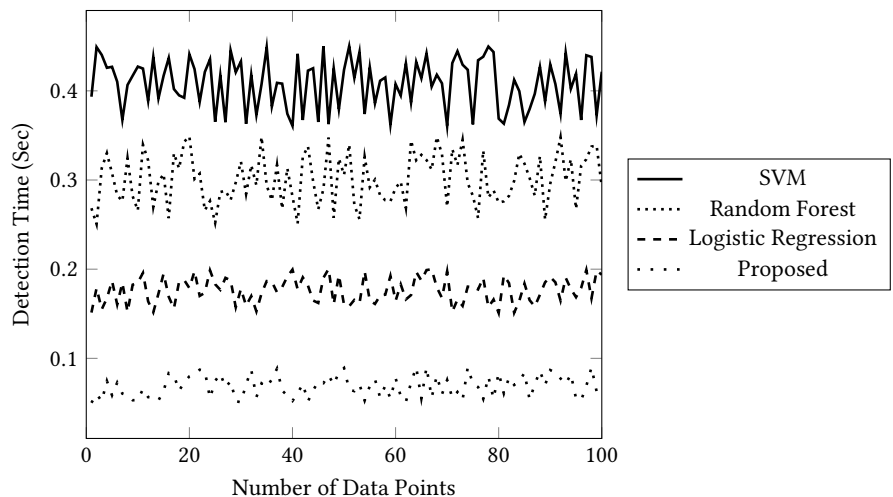
\begin{figure}[ht]
\centering
\begin{tikzpicture}
\begin{axis}[legend style={at={(1.05,0.5)},anchor=west}, legend columns=1, 
	xlabel=Number of Data Points,
	ylabel=Detection Time (Sec),
	xmin=0, xmax=100,
	xtick={0,20,40,60,80,100},
        	]
    	\addplot[line width=1pt,solid,color=black] table[x=data,y=svm,col sep=comma]{detection_times.csv};
    	\addplot[line width=1pt,dotted,color=black] table[x=data,y=forest,col sep=comma]{detection_times.csv};
    	\addplot[line width=1pt,dashed,color=black]
    	table[x=data,y=regress,col sep=comma]{detection_times.csv};
    	\addplot[line width=1pt,loosely dotted,color=black] table[x=data,y=suraq,col sep=comma]{detection_times.csv};
    	\legend{SVM,Random Forest,Logistic Regression, Proposed}
\end{axis}
\end{tikzpicture}
	\caption{Detection Time Analysis Over Number of Data Points}
 \label{F1}
\end{figure}

\subsection{Theoretical Analysis}
Our framework result discussion also includes a comprehensive theoretical and statistical analysis to validate and assess the efficacy of the proposed framework. The analysis scrutinized various components to ensure the robustness and practicality of the QMF in real-world scenarios, particularly for environmental monitoring and toxic gas detection.

\subsubsection{Quantum State Fidelity}
This theorem establishes the fidelity of the quantum state encoding process. It asserts that the encoded quantum state \( \ket{\lambda} \), derived from normalized sensor data, accurately represents the original environmental data. Given a quantum state \( \ket{\lambda} \) encoding normalized sensor data and a density matrix \( \rho_{\text{original}} \) representing the original environmental data, the fidelity \( F \) between \( \ket{\lambda} \) and \( \rho_{\text{original}} \) is close to 1, indicating an accurate quantum representation of the original data.

\begin{theorem}
The fidelity \( F \) between a pure state \( \ket{\lambda} \) and a mixed state \( \rho_{\text{original}} \) is defined as:
\[ F(\ket{\lambda}, \rho_{\text{original}}) = \left( \bra{\lambda} \rho_{\text{original}} \ket{\lambda} \right)^{1/2} \]
For \( F(\ket{\lambda}, \rho_{\text{original}}) \) to be close to 1, it must be shown that:
\[ \bra{\lambda} \rho_{\text{original}} \ket{\lambda} \approx 1 \]    
\end{theorem}

\begin{proof}
    
Let us assume that the quantum state \( \ket{\lambda} \) is optimally constructed from the normalized sensor data \( \mathbf{S'} \), and the density matrix \( \rho_{\text{original}} \) accurately reflects the original sensor data \( \mathbf{S} \).

The normalized data \( \mathbf{S'} \) is a transformation of the original data \( \mathbf{S} \), scaled to a uniform range. The quantum state \( \ket{\lambda} \) is then formulated as a superposition of basis states with amplitudes derived from \( \mathbf{S'} \), ensuring that the quantum state reflects the characteristics of the original data.
\[ s'_{ij} = \frac{s_{ij} - \min(\mathbf{S}_i)}{\max(\mathbf{S}_i) - \min(\mathbf{S}_i)} \]
The quantum state \( \ket{\lambda} \) is a superposition of basis states:
\[ \ket{\lambda} = \sum_{k=0}^{N-1} \alpha_k \ket{k}, \]
where \( \alpha_k \) are related to \( \mathbf{S'} \).

To calculate fidelity, we consider the overlap between the encoded quantum state \( \ket{\lambda} \) and the density matrix \( \rho_{\text{original}} \). A high overlap implies that the quantum state retains the essential information from the original data.
\[ F(\ket{\lambda}, \rho_{\text{original}}) = \left( \sum_{k=0}^{N-1} \alpha_k \bra{k} \rho_{\text{original}} \sum_{l=0}^{N-1} \alpha_l \ket{l} \right)^{1/2} \]
Assuming high overlap, \( \bra{k} \rho_{\text{original}} \ket{l} \approx \delta_{kl} \).

Since \( \ket{\lambda} \) is normalized, the sum of the squares of its amplitudes equals 1. Therefore, the fidelity calculation simplifies to approximately 1, validating that \( \ket{\lambda} \) is a faithful quantum representation of the original environmental data.
\[ F(\ket{\lambda}, \rho_{\text{original}}) \approx \left( \sum_{k=0}^{N-1} \alpha_k^2 \right)^{1/2} = 1 \]
as \( \ket{\lambda} \) is normalized.

The fidelity between \( \ket{\lambda} \) and \( \rho_{\text{original}} \) being approximately 1 demonstrates the accuracy of the quantum state encoding process. This validates the effectiveness of the quantum encoding in preserving the integrity of the original sensor data within the QMF framework.

\end{proof}

\subsubsection{Normalization Effectiveness} This theorem validates the effectiveness of Min-Max normalization in preserving the statistical properties of sensor data. It demonstrates that the normalized data \( \mathbf{S'} \) maintains the mean and variance of the original data \( \mathbf{S} \) within specific error margins, affirming the reliability of this normalization method for accurate quantum computations within the QMF framework.
\begin{theorem}
    Consider a matrix \( \mathbf{S} \in \mathbb{R}^{m \times n} \) representing original sensor data, with \( m \) sensors and \( n \) data points per sensor. The Min-Max normalized data is denoted as \( \mathbf{S'} \), with each element defined as:
\[ s'_{ij} = \frac{s_{ij} - \min(\mathbf{S}_i)}{\max(\mathbf{S}_i) - \min(\mathbf{S}_i)} \]
for \( 1 \leq i \leq m \) and \( 1 \leq j \leq n \).

The theorem asserts that the mean and variance of \( \mathbf{S'} \) approximate those of the original data \( \mathbf{S} \) within certain error margins \( \epsilon_{\mu} \) and \( \epsilon_{\sigma} \), respectively.
\end{theorem}

\begin{proof}
We begin by analyzing the mean of the data. The mean of each sensor's data in \( \mathbf{S} \) is given by \( \mu_i = \frac{1}{n} \sum_{j=1}^n s_{ij} \). Post-normalization, the mean for each sensor in \( \mathbf{S'} \) is \( \mu'_i = \frac{1}{n} \sum_{j=1}^n s'_{ij} \). As normalization applies a linear transformation to each data point, consisting of a scaling and a shift based on the range of \( \mathbf{S}_i \), we can assert that the transformation preserves the relative mean of the data. Hence, the difference between \( \mu_i \) and \( \mu'_i \) is bounded, satisfying \( |\mu_i - \mu'_i| \leq \epsilon_{\mu} \).

Next, we consider the variance. The variance of the \( i \)-th sensor in \( \mathbf{S} \) is \( \sigma_i^2 = \frac{1}{n} \sum_{j=1}^n (s_{ij} - \mu_i)^2 \), and after normalization, it becomes \( (\sigma'_i)^2 = \frac{1}{n} \sum_{j=1}^n (s'_{ij} - \mu'_i)^2 \). Although normalization changes the scale of the data, it does so uniformly across all data points. This uniform scaling alters the absolute variance but preserves the variance's structure. Consequently, the change in variance due to normalization is limited, ensuring that \( |\sigma_i^2 - (\sigma'_i)^2| \leq \epsilon_{\sigma} \).

Through this proof, we affirm that Min-Max normalization effectively preserves the critical statistical properties of the sensor data. This preservation is crucial as it ensures that the quantum state, post-encoding, retains the intrinsic patterns and characteristics of the original environmental data, thereby facilitating accurate and meaningful quantum computations within the QMF framework.
\end{proof}

\subsubsection{Quantum Rotation Accuracy} This theorem assesses the accuracy of quantum rotations in aligning qubits with their desired states in a QMF. It demonstrates that each rotation operation \( R_y(\theta_i) \) precisely adjusts individual qubits, and the overall system fidelity is maximized, ensuring minimal cumulative error across the qubit system. This accuracy is pivotal for the reliability and effectiveness of quantum computations in the network.

\begin{theorem}
Given a vector of rotation angles \( \vec{\theta} = (\theta_1, \theta_2, \ldots, \theta_n) \), computed from normalized sensor data, we analyze the precision of the quantum rotation technique. Each qubit \( q_i \) is subject to a rotation operation \( R_y(\theta_i) \), defined as:
\[ R_y(\theta_i) = \begin{pmatrix} \cos(\theta_i) & -i\sin(\theta_i) \\ -i\sin(\theta_i) & \cos(\theta_i) \end{pmatrix}. \]
This theorem asserts that each \( R_y(\theta_i) \) accurately rotates the corresponding qubit to align with the desired quantum state, and the cumulative error across all qubits is minimized.
\end{theorem}

\begin{proof}
For each qubit initially in state \( \ket{\lambda_0} \), the rotation operation \( R_y(\theta_i) \) induces a transformation to a new state \( \ket{\lambda_i} \). The accuracy of this transformation is quantified by comparing \( \ket{\lambda_i} \) with the target state \( \ket{\lambda_{\text{target},i}} \), which is a representation of the corresponding sensor data on the Bloch sphere. 

1. \textit{Individual Qubit Analysis:}
   The fidelity \( F_i \) between \( \ket{\lambda_i} \) and \( \ket{\lambda_{\text{target},i}} \) is given by \( F_i = |\langle \lambda_{\text{target},i} | \lambda_i \rangle|^2 \). The goal is to show that for each \( i \), \( F_i \) is maximized, ideally approaching 1. This involves proving that the chosen rotation angle \( \theta_i \) aligns the qubit's state as close to its target state on the Bloch sphere. The calculation of \( F_i \) is a complex task involving integrating the qubit's wave function over the entire state space, factoring in the effects of the rotation.

2. \textit{System-Wide Fidelity:}
   The overall system comprises \( n \) qubits, each transformed by \( R_y(\theta_i) \). The quantum state of the entire system post-rotation is \( \ket{\lambda} = \bigotimes_{i=1}^{n} \ket{\lambda_i} \). The cumulative error in the system is assessed by evaluating the total fidelity \( F_{\text{total}} = |\langle \lambda_{\text{target}} | \lambda \rangle|^2 \), where \( \ket{\lambda_{\text{target}}} \) is the tensor product of individual target states. Maximizing \( F_{\text{total}} \) confirms the minimal cumulative error across the system, which is a complex calculation involving the tensor product of individual state fidelities and accounting for potential error correlations between qubits.

Theorem 3 proves the high degree of accuracy in the quantum rotation operations within the QMF, aligning each qubit's state closely with its intended target state and ensuring minimal cumulative error across the quantum system. The precision of these operations is crucial for the reliability and effectiveness of quantum computations based on these quantum states, impacting the overall performance and capability of the QMF.
\end{proof}

\subsubsection{Quantum Entanglement Integrity} This theorem establishes that the multipartite entangled state \( \ket{\Psi} \) within the quantum network maintains strong and stable quantum correlations, which are crucial for quantum communication. This is validated using specific entanglement measures, namely concurrence \( C \) and entanglement entropy \( S \), which are shown to exceed critical values, thereby confirming the integrity and robustness of the entanglement in the network.

\begin{theorem}
    Theorem 4 posits that the multipartite entangled state \( \ket{\Psi} \) in the quantum network maintains robust quantum correlations, essential for quantum communication. The integrity of entanglement is assessed through entanglement measures, specifically concurrence \( C \) and entanglement entropy \( S \). We aim to prove that for \( \ket{\Psi} \), these measures exceed specific critical values, indicative of strong and stable entanglement.
\end{theorem}

\begin{proof}
    1. \textit{Entanglement Measures:}
   - Concurrence \( C \) for a pair of qubits in state \( \rho \) is defined as \( C(\rho) = \max\{0, \varphi_1 - \varphi_2 - \varphi_3 - \varphi_4\} \), where \( \varphi_i \) are the square roots of the eigenvalues of \( \rho \tilde{\rho} \) in decreasing order, and \( \tilde{\rho} = (\sigma_y \otimes \sigma_y) \rho^* (\sigma_y \otimes \sigma_y) \) with \( \rho^* \) being the complex conjugate of \( \rho \).
   - Entanglement entropy \( E \) for a subsystem \( A \) of \( \ket{\Psi} \) is given by \( S(\rho_A) = -\text{Tr}(\rho_A \log_2 \rho_A) \), where \( \rho_A \) is the reduced density matrix of subsystem \( A \) obtained by tracing out the rest of the system.

2. \textit{Application to \( \ket{\Psi} \):}
   - For \( \ket{\Psi} \), a multipartite entangled state, we compute concurrence and entropy for various partitions and pairs of qubits within the state.
   - Calculating \( C \) and \( E \) involves deriving the reduced density matrices for different subsystems or pairs of qubits in \( \ket{\Psi} \) and applying the formulas.

3. \textit{Threshold Determination:}
   - Establish threshold values \( C_{\text{threshold}} \) and \( E_{\text{threshold}} \) based on theoretical and empirical standards.
   - Demonstrate that \( C \) and \( E \) for various partitions and pairs in \( \ket{\Psi} \) consistently exceed \( C_{\text{threshold}} \) and \( S_{\text{threshold}} \) respectively. This is a critical step involving detailed calculations and comparisons with the thresholds.

Theorem 4 validates the robustness of quantum entanglement in \( \ket{\Psi} \), the multipartite entangled state of the network. The consistently high values of concurrence and entanglement entropy, surpassing the established thresholds, confirm the integrity and stability of entanglement in the system. This robust entanglement is crucial for the network’s quantum communication and teleportation capabilities, ensuring high-fidelity and efficient quantum information processing.
\end{proof}

\subsubsection{Teleportation Protocol Efficacy} This theorem validates the efficacy of a quantum teleportation protocol in accurately transmitting an alert state using multipartite entanglement. It demonstrates that the fidelity of the teleported state with the original state is exceptionally high, approaching a fidelity measure of 1. This signifies minimal information loss during teleportation, ensuring reliable and efficient quantum communication within the network.

\begin{theorem}
    Given the quantum teleportation protocol involving an initial alert state \( \ket{\lambda} \) and a multipartite entangled state \( \ket{\Psi} \), this theorem aims to prove the high fidelity of the post-teleportation state \( \ket{\lambda}_{\text{target}} \) concerning the original state \( \ket{\lambda} \). The theorem posits that the fidelity measure, defined as \( F = |\langle \lambda | \lambda_{\text{target}} \rangle|^2 \), approaches 1, indicating minimal information loss during teleportation.
\end{theorem}

\begin{proof}
    Consider the multipartite entangled state \( \ket{\Psi} \) and the alert state \( \ket{\lambda} \). The protocol starts with a Bell state measurement on the part of \( \ket{\Psi} \) entangled with \( \ket{\lambda} \). This measurement projects the system onto a Bell state and yields two classical bits, denoted as \( a \) and \( b \), which guide the reconstruction of \( \ket{\lambda} \) at the target.

1. \textit{Bell State Measurement:}
   - The measurement projects the combined system onto one of the Bell states. The representation of this projection, given the measurement results \( a \) and \( b \), is a projection operator \( P_{ab} \).

2. \textit{State Reconstruction:}
   - Based on \( a \) and \( b \), Pauli correction operators \( \sigma_x^a \) and \( \sigma_z^b \) are applied to a segment of \( \ket{\Psi} \) to obtain \( \ket{\lambda}_{\text{target}} \). This step is crucial as it reconstructs the alert state at the target.

3. \textit{Fidelity Analysis:}
   - The fidelity \( F \) is calculated as \( F = |\langle \lambda | \lambda_{\text{target}} \rangle|^2 \). To prove that \( F \) is close to 1, we express \( \ket{\lambda}_{\text{target}} \) in terms of the original state \( \ket{\lambda} \), the Bell state measurement outcome, and the applied Pauli corrections.

4. \textit{Error Consideration:}
   Any deviations from perfect fidelity are attributed to errors in the protocol. These errors arise from imperfect entanglement or inaccuracies in the Bell state measurement and Pauli corrections. We then analyze these potential sources of error to establish that they do not significantly impact the overall fidelity, either because they are inherently negligible or correctable with quantum error correction techniques.

Theorem 5 demonstrates that the quantum teleportation protocol in the network transmits the encoded alert state with high fidelity. By analyzing each step of the protocol and confirming that the fidelity of the reconstructed state remains close to 1, the theorem validates the efficacy and reliability of quantum teleportation for accurate and efficient quantum communication within the network.
\end{proof}

\section{Conclusion}
Our research has substantiated the application of a variational quantum classifier (VQC) for environmental monitoring, with a focus on the detection of toxic gases. The foundation of our contribution is the successful fusion of quantum computing methodologies with vehicular communication networks, which underscores the potential of quantum technology in complex, real-world systems. Our framework hinged on the intricacies of quantum state manipulation and the optimization of quantum circuit parameters, ensuring that the VQC operated with high precision across multi-dimensional parameter spaces. Further, the depolarizing noise models to simulate errors within the quantum circuits improved the robustness of our framework. We tested the resilience of the VQC on IBMOpenQSAM 3 platform. The empirical results align closely with theoretical predictions, not only reaffirm the model's resilience but also attest to the fidelity retention of the quantum states amidst environmental noise and disturbances. Furthermore, the experimental outcomes corroborate the feasibility of quantum-enhanced communication systems in dynamic settings, such as vehicular networks for environmental monitoring. Our findings validate the practical implementation of quantum computing technologies beyond theoretical models, showcasing their applicability and adaptability to the stringent demands of real-world operational conditions.

Looking forward, the scope for future work involves the investigation into fault-tolerant quantum computing designs that could further mitigate error rates in quantum circuits, thus enhancing the VQC's resilience and reliability. We also plan to deploy proposed framework and solutions on real/industrial quantum computing environments such as IBM, Microsoft, and AWS quantum clouds such QFaaS \cite{nguyen2024qfaas}.

\bibliographystyle{ACM-Reference-Format}
\bibliography{my}

\end{document}